\documentclass[11pt,a4paper]{article}

\usepackage[utf8]{inputenc} 
\usepackage[T1]{fontenc}    %
\usepackage[english]{babel} 
\usepackage[final]{microtype} 

\usepackage{amssymb,amsmath,mathtools} 
\usepackage{graphicx} 
\usepackage{bbm} 
\usepackage{tikz} 

\usepackage[font=small,labelfont=bf]{caption}

\usepackage[hmargin=0.12\paperwidth,vmargin=0.2\paperwidth,bindingoffset=0cm]{geometry}

\usepackage{titlesec}

\titleformat*{\section}{\large\bfseries}
\titleformat*{\subsection}{\bfseries}

\usepackage{amsthm}

\theoremstyle{plain}
  \newtheorem{theorem}{Theorem}[section]
  \newtheorem{prop}[theorem]{Proposition}
  \newtheorem{lemma}[theorem]{Lemma}
  \newtheorem{corollary}[theorem]{Corollary}
\theoremstyle{definition}
  
  \newtheorem{remark}[theorem]{Remark}
\theoremstyle{remark}

\newcommand{\R}{\mathbb{R}}

\newcommand{\E}{\mathbb E}

\renewcommand{\phi}{\varphi} 
\renewcommand{\epsilon}{\varepsilon} 

\DeclarePairedDelimiter{\abs}{\lvert}{\rvert} 

\usepackage{hyperref}   
\hypersetup{
 pdfborder={0 0 0},      
 colorlinks=true,        
 linktoc=page,           
 linkcolor=blue,         
 citecolor=blue,         
}

\title{\bfseries Selberg integral theory \\ and Muttalib--Borodin ensembles}
\author{P.~J.~Forrester and J.~R.~Ipsen\\[1em]%
\small School of Mathematics and Statistics, ARC Centre of Excellence for Mathematical and Statistical Frontiers\\%
\small The University of Melbourne, Victoria 3010, Australia%
}
\date{\today}

\begin{document}	

\maketitle

\begin{abstract}
\noindent
We study Muttalib--Borodin ensembles --- particular eigenvalue PDFs on the half-line --- with classical weights, i.e. Laguerre, Jacobi or Jacobi prime.
We show how the theory of the Selberg integral, involving also Jack and Schur polynomials, naturally leads to a multi-parameter
generalisation of these particular Muttalib--Borodin ensembles, 
and also to the explicit form of underlying biorthogonal polynomials of a single variable. 
A suitable generalisation of the original definition of the Muttalib--Borodin ensemble allows for negative eigenvalues.
In the cases of generalised Gaussian, symmetric Jacobi and Cauchy weights, we show that the
problem of computing the normalisations and the biorthogonal polynomials can be reduced down to Muttalib--Borodin ensembles with
classical weights on the positive half-line.
\end{abstract}


\section{Introduction}

\subsection{Statement of the problem and summary of results}

Our study centres around the 
the family of  probability
density functions (PDFs)
\begin{equation}\label{MB-ensembles}
 {\mathcal P}_{N}(x_1,\ldots,x_N)=
 \frac{1}{Z_{N}}\prod_{1\leq i<j\leq N}(x_j-x_i)(x_j^\theta-x_i^\theta)\prod_{k=1}^Nw(x_k),
 \quad x_k   \in\R^+,
\end{equation}
where $Z_{N}$ is a normalisation constant, $\theta > 0$ is a deformation parameter and $w(x) \ge 0$ is a weight function.
For reasons discussed in the next subsection, these PDFs are referred to as Muttalib--Borodin ensembles.
Our interest is in the case that $w(x)$ is a classical weight. 
We recall (see e.g.~\cite[\S 5.4.3]{Fo10}) that $w(x)$ is referred to as
being classical if its logarithmic derivative is a rational function,
$$
{w'(x) \over w(x)} = - {g(x) \over f(x)},
$$
with $f$ and $g$ are polynomials with no common factors such that their degrees are no bigger than two and one respectively.
Requiring also that $w(x)$ vanishes on its boundary of support, which in turn is contained in $\mathbb R^+$, the only possibilities (up to trivial rescaling)  are
\begin{equation}\label{w1}
w(x) = 
\left\{ \begin{array}{lll} x^a e^{-x}, &x\in\R^+,& {\rm Laguerre} \\
x^a (1 - x)^b, & x \in (0,1), & {\rm Jacobi} \\
x^\alpha/(1 + x)^\beta, &x\in\R^+,& {\rm Jacobi \: prime} \end{array} \right.
\end{equation}
The first two weight in~\eqref{w1} may be transformed as $x \mapsto 1/x$, but this lead to equivalent ensembles, since the products of differences in (\ref{MB-ensembles}) are unchanged by this mapping,
up to a factor of $\prod_{i=1}^Nx_i^p$ for some $p$.
 
It is our aim to relate the ensembles~(\ref{MB-ensembles}) with a classical weight to the theory of the Selberg integral; see \cite{Se44,FW07p} and~\cite[Ch.~4]{Fo10}.
For this purpose one recalls that the Selberg weight refers to the PDF
\begin{equation}\label{S}
{1 \over S_N(\alpha_1,\alpha_2,\tau)}
\prod_{l=1}^N x_l^{\alpha_1 - 1} (1 - x_l)^{\alpha_2 - 1}
\prod_{1 \le j < k \le N} |x_j - x_k|^{2 \tau}, \quad 0 < x_l < 1,
\end{equation} 
where
\begin{eqnarray}
S_N(\alpha_1,\alpha_2,\tau) & := & 
\int_{[0,1]^N} \prod_{i=1}^N x_i^{\alpha_1 - 1} (1 - x_i)^{\alpha_2 - 1}
\prod_{1 \le j < k \le N} |x_j - x_k|^{2 \tau} \, dx_1 \cdots dx_N \nonumber \\
& = & 
 \prod_{j=0}^{N-1} {\Gamma (\alpha_1  + j\tau)
\Gamma (\alpha_2 + j\tau)\Gamma(1+(j+1)\tau) \over
\Gamma (\alpha_1 + \alpha_2  + (N + j-1)\tau) \Gamma (1 + \tau )},
\label{3.2}
\end{eqnarray}
is the Selberg integral \cite{Se44}. We see that in the case $\tau = 1$, this corresponds to the $\theta=1$ Jacobi Muttalib--Borodin ensemble,
i.e.~(\ref{MB-ensembles}) with the Jacobi weight (\ref{w1}), which in turn is the familiar Jacobi unitary ensemble from classical random matrix theory (see e.g.~\cite[Ch.~3]{Fo10}).
To make contact with (\ref{MB-ensembles}) 
for general $\theta > 0$, it is again the $\tau= 1$ case of (\ref{S}) which is relevant, but now augmented by the inclusion of
an extra  Schur polynomial factor (see (\ref{schur}) below for its definition).

 A generalisation of the PDFs~(\ref{MB-ensembles}) permitting negative values is
\begin{equation}\label{MB-ensembles1}
 {\mathcal P}_{N}(x_1,\ldots,x_N)=
 \frac{1}{Z_{N}}\prod_{1\leq i<j\leq N}(x_j-x_i)(({\rm sgn} \, x_j) |x_j|^\theta- ({\rm sgn} \, x_i) | x_i|^\theta)\prod_{k=1}^Nw(x_k),
 \quad x_k\in\R.
\end{equation}
This allows for three further classical (for $c=0$) weight functions
 \begin{equation}\label{w1a}
 w(x) = \left \{ \begin{array}{lll} |x|^{2c} e^{- x^2}, &x\in\R,& {\rm generalised \: Gaussian} \\
 |x|^{2c} (1 - x^2)^\alpha,& x\in(-1,1), & {\rm generalised \: symmetric \: Jacobi} \\ \displaystyle
{|x|^{2c}/(1 + x^2 )^\alpha}, &x\in\R,& {\rm  generalised \: Cauchy}. \end{array} \right.
 \end{equation} 
 A significant feature of these weights in subsequent analysis is that they are all even. 
 It follows that, with a suitable identification of the parameters, the change of variables
 $x^2 = y$ maps the weights~\eqref{w1a} to the weights (\ref{w1}), and consequently from the full real line to the positive half-line.

In Section \ref{S2} we make use of Jack polynomials to extend the Selberg weight to involve $N$ parameters $\{\gamma_i\}_{i=1}^N$.
Specialising to $\tau = 1$, a so-called polynomial ensemble results, which, with the $\gamma_i$ appropriately chosen, gives
 the Jacobi Muttalib--Borodin ensemble for general $\theta > 0$. A limiting case of this gives the Laguerre  Muttalib--Borodin ensemble.
 A different integration formula involving the Selberg weight and Jack polynomials allows for the Jacobi prime  Muttalib--Borodin ensemble 
 to be deduced through analogous working. A practical consequence of this line of theory is that it gives the normalisation in
(\ref{MB-ensembles}) for the weights (\ref{w1}). In Section \ref{S2.5} we show how an underlying parity symmetry of
(\ref{MB-ensembles1}) with the weights (\ref{w1a}) allows the normalisations for this class of Muttalib-Borodin models to be
deduced as corollaries from knowledge of the normalisations for (\ref{MB-ensembles}) with weights (\ref{w1}).

Associated with a Muttalib--Borodin ensemble are two families polynomials in a single variable, $\{p_k(x) \}$ and $\{q_k(x) \}$, 
which satisfy a biorthogonal relation
\begin{equation}\label{biortho-product}
\int w(x)p_k(x)q_\ell(x^\theta) \, dx =h_k\delta_{k\ell},\qquad k,\ell=0,1,\ldots,N-1.
\end{equation}
We show in
Section \ref{S3} how, for the weights  (\ref{w1}) and  (\ref{w1a}), Selberg integral theory can be used to give the explicit form
of these polynomials. For the weights (\ref{w1}), these polynomials can be specified as the averaged characteristic polynomial for 
(\ref{MB-ensembles}) with $N = k$, and the averaged characteristic polynomial of the underlying matrix raised to the
power of $\theta$ respectively. The latter is the most straightforward to compute using Selberg integral theory. To compute the
polynomials $\{p_k(x) \}$,
we make use of the multi-parameter generalisation of the Muttalib--Borodin ensemble introduced in Section \ref{S2} 
to obtain a structured expression for  the averaged characteristic polynomial  and then specialise the parameters.
The biorthogonal polynomials for the weights (\ref{w1a}) follows as a corollary of knowledge of the
 biorthogonal polynomials for the weights (\ref{w1}).

\subsection{Context of the Muttalib--Borodin model in random matrix theory}

The PDFs (\ref{MB-ensembles})
 were introduced
into random matrix theory by Muttalib~\cite{Mu95}. When the deformation parameter equals unity ($\theta=1$),
(\ref{MB-ensembles}) is the functional form of the eigenvalue PDF for an ensemble $\{H\}$ of positive definite Hermitian matrices
distributed according to the PDF on the matrices $e^{-{\rm Tr} \, V(H)}$ and, hence, $w(x) = e^{- V(x)}$ (see e.g.~\cite[Ch.~1 and 3]{Fo10}).
Muttalib's interest was in the limit $\theta \to 0^+$ when the PDF (\ref{MB-ensembles}) corresponds to a 
simplification of the eigenvalue PDF found by Beenakker and Rajei~\cite{BR93,BR94} one year earlier
in their exact solution of the so-called DMPK equation (Dorokhov--Mello--Pereyra--Kumar~\cite{MPK88}) in the theory of quantum conductance.

It is questionable whether Muttalib's ansatz  (\ref{MB-ensembles})
allows quantitative predictions in the context where they were proposed. On the other hand, the $\theta = 2$ case of
this same functional form
subsequently appeared in the study of other physical models, such as disordered bosons~\cite{LSZ06} (with $w(x) = e^{-x}$) and matter coupled to two-dimensional quantum gravity~\cite{EZ92,EK95a}. For general $\theta > 0$ 
and $w(x) = e^{-x}$, (\ref{MB-ensembles}) was shown to result from a model of last passage percolation based on the
Robinson-Knuth-Schensted correspondence \cite{FR02b}.
Borodin~\cite{Bor99} made considerable progress on the mathematical aspects
of the correlation functions associated with the PDF~(\ref{MB-ensembles}).
In particular, he showed that  the local scaling regime near the hard edge of these ensembles with classical weight functions 
(see below for an explanation of this notion)
gave rise to new correlation kernels involving Wright's Bessel function, generalising the classical Bessel kernel \cite{Fo93a}.
This exit from ordinary random matrix statistics, of course, added to the interest in such ensembles, at least from a theoretical
viewpoint.

 Recently, the PDFs (\ref{MB-ensembles}) have drawn renewed attention in the random matrix literature~\cite{CR14,KS14,Ch14,FL14,FW15,FLZ15} with the name ``Muttalib--Borodin ensembles'' coined in~\cite{FW15}. This increase in attention is justified by newly discovered links to several applications, for example to triangular random matrix models~\cite{Ch14} and to certain combinatorial numbers~\cite{FL14}. Prominently, there also exists an intimate connection to products of random matrices and free probability when either $\theta$ or $1/\theta$ is an integer \cite{KS14}.

\section{Selberg integral theory and normalisation}
\label{S2}

In this section, we briefly recall how the theory of symmetric polynomials is related to the Selberg integral theory. 
Using this relation, we compute the normalisations for all classical Muttalib--Borodin ensembles.

\subsection{Jack polynomials and probability density functions}
\label{sec:general/JS}

The $N$ variables in the Selberg weight (\ref{S}) is only coupled through the product of differences
\[
\prod_{1 \leq j < k \leq N} |x_j - x_k|^{2 \tau}.
\]
Associated with this product is a family of multi-variable symmetric polynomials $P_\kappa^{(1/\tau)}(t_1,\dots, t_N)$, referred to
as Jack polynomials (see e.g.~\cite[Ch.~12]{Fo10}, \cite{KK09}). Here $\kappa = (\kappa_1,\dots, \kappa_N)$ is an ordered array
of non-negative integers $\kappa_1 \ge \kappa_2 \ge \cdots \ge \kappa_N$, or equivalently a partition with no more than $N$ parts.
These polynomials have the structure
$$
P_\kappa^{(1/\tau)}(x_1,\dots, x_N) = m_\kappa + \sum_{\mu < \kappa} c_{\kappa, \mu} m_\mu,
$$
where $m_\rho = m_\rho(x_1,\dots, x_N)$ is the monomial symmetric polynomial indexed by $\rho$ (i.e.~the symmetrisation of
$\prod_{i=1}^N x_i^{\kappa_i}$ appropriately normalised), the $c_{\kappa, \mu}$ are expansion coefficients that depend on
$\tau$, and $\mu < \kappa$ is the partial ordering on partitions specified by the requirement that $\sum_{i=1}^m \mu_i \le
\sum_{i=1}^m \kappa_i$ for each $m=1,\dots,N$. 

One way the Jack polynomials are related to the product of differences is through the orthogonality
$$
\int_{{\mathcal C}^N} {dz_1 \over 2 \pi i\, z_1} \cdots {dz_N \over 2 \pi i\, z_N} \,
\prod_{1 \le j < k \le N} | z_k - z_j|^{2 \tau} P_\kappa^{(1/\tau)}(z_1,\dots, z_N) P_\nu^{(1/\tau)}(\bar{z}_1,\dots, \bar{z}_N)
\propto \delta_{\kappa, \nu},
$$
where $\mathcal C$ is the unit circle about the origin in the complex plane. 
More significant for our present purposes is the integration formula, conjectured by MacDonald \cite{Ma87}, and
subsequently proved by Kadell~\cite{Ka97g} and Kaneko \cite{Ka93} (see \cite[Eq.~(12.143)]{Fo10})
\begin{multline}\label{MKK}
{1 \over S_N(\alpha_1,\alpha_2,\tau)}\int_{{(0,1)^N}}dx_1\cdots dx_N \,\prod_{l=1}^N x_l^{\alpha_1-1}(1-x_l)^{\alpha_2-1} 
 \prod_{{1 \le j < k \le N}} |x_j - x_k|^{2 \tau}
{P_\lambda^{(1/\tau)}(x_1,\ldots,x_N)} \\
=\frac{(\alpha_1+(N-1)\tau)_\lambda^{(1/\tau)}}{(\alpha_1 + \alpha_2  +(N-2)\tau)_\lambda^{(1/\tau)}}P_\lambda^{(1/\tau)}((1)^N),
\end{multline}
where
\begin{equation}\label{pochhammer}
(a)^{(1/\gamma)}_\lambda:=\prod_{k=1}^N\frac{\Gamma(a-(k-1)\gamma+\lambda_k)}{\Gamma(a-(k-1)\gamma)},
\end{equation}
denotes a generalised Pochhammer symbol. The notation $(1)^N$ in the last factor of (\ref{MKK})
denotes $1$ repeated $N$ times. A product formula
for $P_\lambda^{(1/\tau)}((1)^N)$ involving $\tau$ and the parts of $\lambda$ is known; see e.g.~\cite[Eq.~(12.105)]{Fo10}.

It follows from the combinatorial expression for the Jack polynomials involving a sum over semi-standard tableaux \cite{KS97}
that the Jack polynomials are positive for the $x_i$ positive. Hence the integration formula (\ref{MKK}) can be used to 
specify a normalised PDF supported on $(0,1)^N$
\begin{equation}\label{MKK1}
{1 \over {\cal N}_\lambda(\alpha_1,\alpha_2,\tau)}
 \prod_{l=1}^N x_l^{\alpha_1 - 1} (1 - x_l)^{\alpha_2 - 1} 
\prod_{1 \le j < k \le N} |x_j - x_k|^{2 \tau}
{P_\lambda^{(1/\tau)}(x_1,\ldots,x_N)},
\end{equation}
where
\begin{equation}\label{MKK2}
 {\cal N}_\lambda(\alpha_1,\alpha_2,\tau) = S_N(\alpha_1,\alpha_2,\tau) \frac{(\alpha_1+(N-1)\tau+1)_\lambda^{(1/\tau)}}{(\alpha_1 + \alpha_2  +(N-2)\tau+2)_\lambda^{(1/\tau)}}P_\lambda^{(1/\tau)}((1)^N).
\end{equation} 
This is the viewpoint that relates to the Jacobi Muttalib--Borodin ensemble, as we will proceed to demonstrate.

\subsection{Schur polynomials and the Jacobi Muttalib--Borodin ensemble}

We now specialise to $\tau = 1$. We then have that the Jack polynomials can be identified as the Schur
polynomials $s_\lambda$, which in turn can be expressed as a ratio of determinants (see e.g.~\cite[Prop.~10.1.5]{Fo10}),
\begin{equation}\label{schur}
P_\lambda^{(1)}(x_1,\ldots,x_N) = s_\lambda(x_1,\ldots,x_N)
=\frac{\det\big[x_j^{N-i+\lambda_i}\big]_{1\leq i,j\leq N}}{\det\big[x_j^{N-i}\big]_{1\leq i,j\leq N}}.
\end{equation}
The determinant in the denominator is the Vandermonde determinant, with the well known evaluation 
\begin{equation}\label{Van}
\det\big[x_j^{N-i}\big]_{1\leq i,j\leq N} = \prod_{1 \le i < j \le N} (x_i - x_j)
\end{equation}
(see e.g.~\cite[Eq.~(1.173)]{Fo10}). 
The PDF (\ref{MKK1}) then specialises to a so-called polynomial ensemble \cite{KS14}, meaning that it is of
the form
\begin{equation}\label{Kg}
{1 \over C_N} \prod_{1 \le i < j \le N} (x_i - x_j ) \det [ g_j(x_k) ]_{j,k=1,\dots,N}
\end{equation} 
for certain $\{g_j(x)\}$.

\begin{prop}\label{P1}
The functional form
\begin{equation}\label{MKK3}
{1 \over {\cal N}_\lambda^{\rm J}(\alpha_1,\alpha_2,1)}
 \prod_{l=1}^N x_l^{\alpha_1 - 1} (1 - x_l)^{\alpha_2 - 1} 
\prod_{1 \le i < j \le N} (x_i - x_j )  \det \big[x_j^{N-i+\lambda_i}\big]_{1\leq i,j\leq N},
\end{equation}
where
\begin{equation}\label{MKK4}
 {\cal N}_\lambda^{\rm J}(\alpha_1,\alpha_2,1)  =
 \prod_{k=1}^N {\Gamma(\alpha_1 + k - 1 + \lambda_{N-k+1}) \Gamma(\alpha_2 + k - 1) \over
 \Gamma(\alpha_1 + \alpha_2 + N + k - 2 + \lambda_{N-k+1})} \prod_{1 \le i < j \le N} (\lambda_i - \lambda_j + j - i)
\end{equation}
is a PDF on  $(0,1)^N$.
\end{prop}

\begin{proof}
Substituting (\ref{schur}) with the denominator simplified according to the Vandermonde determinant evaluation (\ref{Van})
gives (\ref{MKK3}), but with normalisation
\begin{equation}\label{Sx}
 {\cal N}_\lambda^{\rm J}(\alpha_1,\alpha_2,1)  = S_N(\alpha_1,\alpha_2,1) \frac{(\alpha_1+ N-1)_\lambda^{(1)}}{(\alpha_1 + \alpha_2  + 2N-2)_\lambda^{(1)}}.
 \end{equation}
The expression (\ref{MKK4}) follows from this by making use of the formula
(see e.g.~\cite[Eq.~(10.23)]{Fo10})
\begin{equation}\label{Sx1}
s_\lambda((1)^N) =  \prod_{1 \le i < j \le N} {\lambda_i - \lambda_j + j - i \over j - i}
\end{equation}
together with the generalised Pochhammer symbol~(\ref{pochhammer}) and the Selberg integral (\ref{3.2}).
\end{proof}

\begin{remark}
In the derivation of (\ref{MKK3}), the parameters $\{\lambda_i\}$ are non-negative integers. However, an application of Carlson's theorem
(see e.g.~\cite[Prop.~4.1.4]{Fo10}) shows that the value of the normalisation remains valid for continuous values of
$\{\lambda_i\}$.
\end{remark}

It only remains to extract the normalisation for the Jacobi Muttalib--Borodin ensemble from Proposition~\ref{P1} by choosing the parameters correctly. 
Write
\begin{equation}\label{Na}
\gamma_i=N - i + \lambda_i + \alpha_1 - 1,
\end{equation}
then (\ref{MKK3}) and (\ref{MKK4}) reads
\begin{equation}\label{MKK5}
{1 \over N!\,C_N^{\rm J}} \prod_{1 \le i < j \le N}{1 \over  (\gamma_i - \gamma_j)}
\prod_{l=1}^N (1 - x_l)^{\alpha_2-1} 
\prod_{1 \le i < j \le N} (x_i - x_j)
\det [ x_j^{\gamma_i} ]_{i,j=1,\dots,N}
\end{equation}
where
\begin{equation}\label{MKK6}
C_N^{\rm J} = \prod_{l=1}^N {\Gamma(\gamma_l + 1) \Gamma(l-1+\alpha_2) \over \Gamma(\gamma_l + N + \alpha_2)}.
\end{equation}
 An explicit random matrix ensemble with eigenvalue PDF (\ref{MKK5}) has recently been identified. To specify this ensemble,
 let $Y = [y_{j,k}]_{j,k=1,\dots,N}$ and $Z = [z_{j,k}]_{j,k=1,\dots,N}$ be upper-triangular random matrices with all non-zero entries
 independent. Let the strictly upper triangular entries be distributed as standard Gaussians and the diagonal entries be real
 and positive and given in terms of gamma distributions by the requirement that
 $| y_{k,k}|^2 = \Gamma[\gamma_k + 1, 1]$ and  $| z_{k,k}|^2 = \Gamma[\beta_k + 1, 1]$. Then according to
 \cite[Corollary 2.14 with $n=N$]{Fo14} the eigenvalue PDF of the random matrix $(\mathbb I_N +  Z^\dagger Z (Y^\dagger Y)^{-1})^{-1}$
 is precisely (\ref{MKK5}).
 
 As noted in  \cite[Corollary 2.14]{Fo14}, the PDF (\ref{MKK5}) contains the Jacobi Muttalib--Borodin model as a special case.
 Thus set $\gamma_j = \theta(j-1) + c$. The determinant can then be evaluated as a product since it is an example of a
 Vandermonde determinant, and (\ref{MKK5})  reduces to  (\ref{MB-ensembles}) with the Jacobi weight from
 (\ref{w1}), parameters $a=c$ and $b = \alpha_2-1$. Reinstating the parameters $a,b$, the corresponding normalisation of Jacobi Muttalib--Borodin ensemble is
\begin{equation}\label{Zp}
 {1 \over Z_N^{\rm J}} =
 {1 \over \theta^{N(N-1)/2}}
 \prod_{l=1}^N { \Gamma(\theta(l-1) + N+ a + b + 1)  \over  \Gamma(\theta(l-1) + a + 1) \Gamma(l-1+ b + 1)  \Gamma(l+1)}.
\end{equation} 
Note that we have $Z_N^{\rm J}|_{\theta = 1} = S_N(a,b,1)$ as required.

\subsection{Laguerre Muttalib--Borodin ensemble}

It is well known (see e.g.~\cite[\S 4.7.1]{Fo10}) that the Selberg density (\ref{S}), upon the scaling $x_l \mapsto x_l/\alpha_2$,
and in the limit $\alpha_2 \to \infty$ reduces to the Laguerre Selberg density
\begin{equation}\label{MKK7}
{1 \over W_N(\alpha_1,\tau)} \prod_{l=1}^N x_l^{\alpha_1-1} e^{- x_l} 
\prod_{1 \le j < k \le N} |x_j - x_k|^{2 \tau}, \quad x_l \in \mathbb R^+
\end{equation} 
where
$$
 W_N(\alpha_1,\tau) = \prod_{j=0}^{N-1} {\Gamma(1 + (j+1) \tau) \Gamma(\alpha_1 + j \tau) \over \Gamma(1 + \tau)}.
 $$
 This same scaling and limiting procedure can be applied to the PDF of Proposition \ref{P1}.
 
 \begin{corollary}
The functional form
\begin{equation}\label{MKK3a}
{1 \over N!} {1 \over {\cal N}_\lambda^{\rm L}(\alpha_1)}
 \prod_{l=1}^N x_l^{\alpha_1 - 1} e^{- x_l}
\prod_{1 \le i < j \le N} (x_i - x_j )  \det \big[x_j^{N-i+\lambda_i}\big]_{1\leq i,j\leq N},
\end{equation}
where
\begin{equation}\label{MKK4a}
 {\cal N}_\lambda^{\rm L}(\alpha_1)  =
 \prod_{k=1}^N \Gamma(\alpha_1 + k - 1 + \lambda_k)  \prod_{1 \le i < j \le N} (\lambda_i - \lambda_j + j - i)
\end{equation}
is a PDF on  $(\R^+)^N$.
\end{corollary}

Substituting for $\lambda_i$ in favour of $\gamma_i$ according to  (\ref{Na}) gives the PDF
\begin{equation}\label{MKK5a}
{1 \over N!\,C_N^{\rm L}} \prod_{1 \le i < j \le N}{1 \over  (\gamma_i - \gamma_j)}
\prod_{l=1}^N e^{-x_l}
\prod_{1 \le i < j \le N} (x_i - x_j)
\det [ x_j^{\gamma_i} ]_{i,j=1,\dots,N},
\end{equation}
where
\[
C_N^{\rm L}=\prod_{l=1}^N \Gamma(\gamma_l + 1)
\]
This was first isolated in the work of Cheliotis \cite{Ch14} as the eigenvalue PDF of the random matrix $Y^\dagger Y$, where
$Y$ is the random upper triangular matrix specified below (\ref{MKK5}).
Setting $\gamma_j = \theta(j-1) + c$ shows that 
the determinant can then be evaluated as a product since it is an example of a
 Vandermonde determinant, and (\ref{MKK5a})  reduces to  (\ref{MB-ensembles}) with the Laguerre weight from
 (\ref{w1}), parameter $a=c$, as already noticed in \cite{Ch14}. In particular, after reinstating the parameter $a$, we have for the
normalisation Laguerre Muttalib--Borodin ensemble
 \begin{equation}\label{ZL}
 {1 \over Z_N^{\rm L}} =  {1 \over \theta^{N(N-1)/2}}
  {1 \over \prod_{l=1}^N \Gamma(\theta(l-1) + a + 1) \Gamma(l+1)}.
\end{equation}
Again note that we have $Z_N^{\rm L}|_{\theta = 1} = W_N(a-1,1)$ as required.

\subsection{Jacobi prime Muttalib--Borodin ensemble}
The beta density refers to the PDF supported on $(0,1)$
$$
{\Gamma(\alpha + \beta) \over \Gamma(\alpha) \Gamma(\beta)} x^{\alpha - 1} (1 - x)^{\beta - 1}.
$$
Changing variables $y = x/(1 - x)$ gives the PDF supported on $(0,\infty)$,
$$
{\Gamma(\alpha + \beta) \over \Gamma(\alpha) \Gamma(\beta)} y^{\alpha - 1} (1 + y)^{-\alpha - \beta},
$$
referred to as the beta prime distribution \cite{WikiBp}. Since the functional form in the first is the Jacobi weight, we refer to the functional
form in the second as the Jacobi prime weight; recall (\ref{w1}). Only in the case $\theta = 1$ of  (\ref{MB-ensembles}) does this
change of variables leave the products of differences unchanged, after factorisation of terms not involving the coupling of the
variables. Nonetheless a generalisation of (\ref{MKK}) involving the Jacobi prime weight is known in the literature \cite{Wa05,FS09}. In the case
$\tau = 1$ it reads
\begin{align}\label{SS}
&{1 \over S_N(a+1,b+1,1)} \int_0^\infty  dx_1 \cdots \int_0^\infty dx_N \prod_{l=1}^N {x_l^{a} \over  (1 + x_l)^{a + b + 2N } }
\prod_{1 \le i < j \le N} (x_i - x_j)^2
s_\lambda(x_1,\ldots,x_N) \nonumber \\
&\qquad=\frac{(a+N)^{(1)}_\lambda}{(-1)^{\abs\lambda}(-b)^{(1)}_\lambda} s_\lambda((1)^N) 
= \prod_{k=1}^N\frac{\Gamma(a+N - k +\lambda_k)}{\Gamma(a+k)}\frac{\Gamma(b+k-\lambda_k)}{\Gamma(b+k)} s_\lambda((1)^N).
\end{align}
We remark that a special case of (\ref{SS}) can also be found in \cite{FK07}.

Inserting the determinant formula (\ref{schur}) for $s_\lambda(x_1,\dots,x_N)$ and 
product formula  for $s_\lambda((1)^N)$ (\ref{Sx1}), we conclude from (\ref{SS}) the analogue of Proposition
\ref{P1}.

\begin{prop}\label{P2}
The functional form
\begin{equation}\label{MKK3x}
{1 \over N!} {1 \over {\cal N}_\lambda^{\rm Jp}(a,b)}
 \prod_{l=1}^N {x_l^{a}  \over (1 + x_l)^{a+b+2N} }
\prod_{1 \le i < j \le N} (x_i - x_j )  \det \big[x_j^{N-i+\lambda_i}\big]_{1\leq i,j\leq N},
\end{equation}
where
\begin{equation}\label{MKK4x}
 {\cal N}_\lambda^{\rm Jp}(a,b) =
 \prod_{k=1}^N {\Gamma(a + N +1 - k  + \lambda_k) \Gamma(b + k - \lambda_k) \over
 \Gamma(a+b + N + k )} \prod_{1 \le i < j \le N} (\lambda_i - \lambda_j + j - i)
\end{equation}
is a PDF on  $(\mathbb R^+)^N$. This is well defined for $\lambda_1 < b + 1$.
\end{prop}

We now write $a+N-i+\lambda_i = \gamma_i$ and $a+b = d$ to read off from (\ref{MKK3x}) and (\ref{MKK4x}) the PDF
\begin{equation}\label{MKK5x}
{1 \over N!\,C_N^{\rm Jp}} \prod_{1 \le i < j \le N}{1 \over  (\gamma_i - \gamma_j)}
\prod_{l=1}^N {1 \over (1 + x_l)^{d + 2N} }
\prod_{1 \le i < j \le N} (x_i - x_j)
\det [ x_j^{\gamma_i} ]_{i,j=1,\dots,N},
\end{equation}
where
\begin{equation}\label{MKK6x}
C_N^{\rm Jp} =  \prod_{k=1}^N {\Gamma(1  + \gamma_k) \Gamma(d + N  - \gamma_k) \over
 \Gamma(d + N + k )}.
\end{equation}
Setting $\gamma_j = \theta(j-1) + c$ shows that 
 (\ref{MKK5x})  reduces to  (\ref{MB-ensembles}) with the Jacobi prime weight from
 (\ref{w1}), parameter $\alpha=c$, $\beta = d + 2N$ and normalisation in terms of the original parameters
 \begin{equation}\label{ZLa}
 {1 \over Z_N^{\rm Jp}} =  {1 \over \theta^{N(N-1)/2}}
  \prod_{k=1}^N { \Gamma(\beta - N + k) \over 
  \Gamma(\theta(k-1) + \alpha + 1 )
  \Gamma(\beta  - \alpha  - N - \theta (k-1)) \Gamma(k+1)}.
\end{equation}

\subsection{Muttalib--Borodin ensembles allowing negative eigenvalues}
\label{S2.5}

We now turn our attention to calculating the normalisation for the Muttalib--Borodin ensembles supported on the full real line, i.e. PDFs (\ref{MB-ensembles1}) with the weights (\ref{w1a}). The tractability of this task relies on exploiting an underlying parity
symmetry. In the case $\theta = 1$, when (\ref{w1}) is the PDF for an Hermitian matrix ensemble with a unitary symmetry, this
property was previously identified in \cite{Fo06}.

\begin{prop}\label{PA1}
Augment the notation $Z_N$ for the normalisations in (\ref{MB-ensembles}) and (\ref{MB-ensembles1}) to include the weight
$w(x)$ by writing $Z_N[w(x)]$, and further write $Z_N^+[w(x)]$ for (\ref{MB-ensembles})  when all eigenvalues are positive.
Let 
\begin{equation}\label{N1}
N_1 := \lfloor(N+1)/2\rfloor, \qquad N_2 = \lfloor N/2\rfloor,
\end{equation}
and suppose $w(x)$ is an even function of $x$. We have
\begin{equation}\label{zw1}
Z_N[w(x)] =  {N! \over N_1! N_2!} Z_{N_1}^+[x^{-1/2} w(x^{1/2})]  Z_{N_2}^+[x^{\theta/2} w(x^{\theta /2})] .
\end{equation}
In particular
\begin{align}\label{zw2} 
Z_N[|x|^{2c} e^{- x^2}]  & =  Z_{N_1}^+[x^{c-1/2} e^{-x}]    Z_{N_2}^+[x^{c+\theta/2} e^{-x}]  \nonumber \\[.2em]
Z_N[|x|^{2c} (1 - x^2)^a]   & =  Z_{N_1}^+[x^{c-1/2} (1 - x)^a ]    Z_{N_2}^+[x^{c+\theta/2} (1 - x)^a ]  \nonumber \\[.2em]
Z_N [{|x|^{2c}/(1+x^2)^a}] &= Z_{N_1}^+[{x^{c-1/2}/(1+x)^a}] Z_{N_2}^+[{x^{c+\theta/2}/(1+x)^a)}]. 
\end{align}
\end{prop}

\begin{proof}
Applying (\ref{Van}) to write the two products of differences in (\ref{MB-ensembles1}) as Vandermonde
determinants, it is then a standard
calculation (see e.g.~\cite[Proof of Prop.~5.2.1]{Fo10}) to express $Z_N[w(x)]$ as a single determinant,
\begin{equation}\label{VM1}
Z_N[w(x)] = N! \det \bigg [ \int_{-\infty}^\infty w(x) |x|^{j-1+\theta(k-1)} ( {\rm sgn} \, (x))^{j+k} \, dx \bigg ]_{j,k=1,\dots,N}.
\end{equation}
Under the assumption that $w(x)$ is even, the matrix in (\ref{VM1}) has a chequerboard structure, with elements in positions
$(j,k)$ with $j+k$ odd equal to zero. Rearranging the order of rows and columns shows that the determinant can be written
in a block (and therefore factored) form
\begin{equation}\label{VM2}
\det \begin{bmatrix} A & 0_{N_1 \times N_2} \\
 0_{N_2 \times N_1} & B \end{bmatrix} = \det A \det B,
 \end{equation}
 where $N_1, N_2$ are specified by (\ref{N1}), and 
$$
 A   = \bigg [ \int_{-\infty}^\infty w(x) |x|^{2j-2+\theta(2k-2)}\, dx \bigg ]_{j,k=1}^{N_1}, \quad
 B   =  \bigg [ \int_{-\infty}^\infty w(x) |x|^{2j-1+\theta(2k-1)} \, dx \bigg ]_{j,k=1}^{N_2}.  
$$
Consider the matrix $A$. We change variables $x^2 = y$ in the integrand so that the general
element reads $\int_0^\infty y^{-1/2} w(y^{1/2}) y^{j-1 + \theta (k-1)} \, dy$. On the other hand,
the application of analogous working which led to (\ref{VM1}) shows if $u(x)$ is a weight with support on the positive half-line, then
$$
Z_n[u(x)] = n! \det \bigg [ \int_0^\infty u(x) x^{j-1 + \theta(k-1)} \, dx \bigg ]_{j,k=1}^n.
$$
Thus, by comparison we have $A = {1 \over N_1!} Z_{N_1}^+[x^{-1/2} w(x^{1/2}) ]$. Similarly, $B = {1 \over N_2!} Z_{N_2}^+[x^{\theta/2} w(x^{1/2})]$.
 Substituting in (\ref{VM2})  gives  (\ref{zw1}). 
 \end{proof}
 
The normalisations of the ensembles~\eqref{MB-ensembles1} with weights~\eqref{w1a} readily follows from Proposition~\ref{PA1} together with previously found normalisations:~\eqref{Zp},~\eqref{ZL}, and~\eqref{ZLa}.

\section{Selberg integral theory and biorthogonal systems}
\label{S3}

Consider two families of monic polynomials, $\{p_k(x) \}$ and $\{q_k(x)\}$. By adding appropriate columns
in the Vandermonde determinant (\ref{Van}), we see that for the products of differences in (\ref{MB-ensembles})
we can write
\[
\prod_{1 \le j < k \le N} (x_k - x_j) = \det[ p_{j-1}(x_k) ]_{j,k=1}^{N} , 
\qquad\text{and}\qquad
\prod_{1 \le j < k \le N} (x_k^\theta - x_j^\theta)
= \det[ q_{j-1}(x_k^\theta) ]_{j,k=1}^{N}.
\]
It is a standard result in random matrix theory (see e.g.~\cite[\S 5.8]{Fo10}) that by choosing $\{p_k(x)\}$ and
$\{q_k(x)\}$ to have the biorthogonal property~\eqref{biortho-product}, 
the   $n$-point correlation function has the evaluation
\begin{equation}\label{correlation-det}
R^{N}_n(x_1,\ldots,x_n)=\det_{1\leq i,j \leq n}\Bigg[\sum_{k=0}^{N-1}\frac{p_k(x_i)q_k(x^\theta_j)}{h_k}\Bigg].
\end{equation}
It is also an easy result to show that the normalisation $Z_N$ in (\ref{MB-ensembles}) can be expressed in terms
of the normalisations $\{h_k\}$ for the biorthogonal polynomials according to
\begin{equation}\label{normalisation}
Z_N=N!\prod_{k=0}^{N-1}h_k.
\end{equation}

The study of the biorthogonal polynomials associated with the  Jacobi  and Laguerre weights can be traced back to
Didon \cite{Di69} and Deruyts \cite{De86}; these references were brought to modern day attention in
\cite{AV83}. In particular, the explicit series expansion was given for the polynomials
$\{q_k(x)\}$, and a Rodrigues type formula given for the polynomials $\{p_k(x)\}$. In the Laguerre case, unaware of these
earlier works, Konhauser \cite{Ko67} rediscovered these formulae, obtained the normalisation, and also gave integral representations of both families of polynomials,
which are now sometimes referred to as the Konhauser biorthogonal polynomials. Our aim in this section is to demonstrate how
Selberg integral theory relates to the determination of the biorthogonal systems for the weights (\ref{w1}) and  (\ref{w1a}).
In doing so, we will use the fact that the monic biorthogonal polynomials corresponding to the PDF~(\ref{MB-ensembles}) can be written in terms of the averages with respect to the same PDF for $N=k$
\begin{equation}\label{47}
p_k(x)=\E\bigg[\prod_{l=1}^k(x-x_l)\bigg], \qquad
q_k(x)=\E\bigg[\prod_{l=1}^k(x-x_l^\theta)\bigg].
\end{equation}
For $\theta=1$, this is the well-known Heine formula (see e.g.~\cite[Prop.~5.1.4]{Fo10}). The generalisation of the Heine formula to all $\theta>0$ is trivial and has already been made explicit in \cite[Eq.~(3.3)]{FW15}.

\subsection{Normalisations}

For the weights (\ref{w1}), the normalisations $\{h_k\}$ follow from (\ref{Zp}), (\ref{ZL}), (\ref{ZLa}) as a consequence of
(\ref{normalisation}).

\begin{prop}
We have
\begin{align*}
h_k^{\rm J} & = {\Gamma(\theta k + a + 1) \Gamma(k+1) \Gamma(k +b + 1) \Gamma((k + a + b + 1)/\theta) \over
\Gamma((\theta + 1)k + a + b + 2) \Gamma(k+1+(k + a + b + 1)/\theta)}  \\[.2em]
h_k^{\rm L} & = \theta^k \Gamma(\theta k + a + 1) \Gamma(k+1) \\[.2em]
h_k^{\rm Jp} & = { \Gamma(\theta k + \alpha + 1) \Gamma(k+1) 
\Gamma(\beta - \alpha -1- (\theta+1) k)  \Gamma((\beta - \alpha - k -1)/\theta - k +1) 
 \over  \Gamma((\beta - \alpha - k - 1)/\theta +1) \Gamma(\beta - k )}.
\end{align*}
\end{prop}

\begin{proof}
We know from~\eqref{normalisation} that $h_k=Z_{k+1}/((k+1)Z_k)$. Thus, for the Jacobi case, it follows from~\eqref{Zp} that
\begin{equation}
h_k^{\rm J}=\theta^k\frac{\Gamma(\theta k+a+1)\Gamma(k+b+1)\Gamma(k+1)}{\Gamma((\theta+1)k+a+b+2)}
\prod_{l=1}^k\frac{\Gamma(\theta(l-1)+k+a+b+1)}{\Gamma(\theta(l-1)+k+a+b+2)}.
\end{equation}
Moreover, using the recursive property of the gamma function, $\Gamma(z+1)=z\Gamma(z)$, we see that
\begin{equation}
\theta^k\prod_{l=1}^k\frac{\Gamma(\theta(l-1)+k+a+b+1)}{\Gamma(\theta(l-1)+k+a+b+2)}=
\prod_{l=1}^k\frac{1}{l-1+\frac{k+a+b+1}{\theta}}
=\frac{\Gamma(\frac{k+a+b+1}{\theta})}{\Gamma(k+1+\frac{k+a+b+1}{\theta})},
\end{equation}
which proves the proposition in the Jacobi case. The Laguerre and Jacobi prime cases follows in a similar manner.
\end{proof}

In the case of the Muttalib--Borodin ensembles allowing for negative eigenvalues (\ref{MB-ensembles1}), provided the
weight function $w(x)$ is even, we can choose $p_k(x)$ and $q_k(x)$ to be even (odd) for $k$ even (odd).   Thus we can
write
\begin{equation}\label{pq}
p_{2k+1}(x) = x\, \tilde{p}_k(x^2), \quad q_{2k+1}(x) = x\, \tilde{q}_k(x^2), \quad
p_{2k}(x) =  \tilde{P}_k(x^2), \quad q_{2k}(x) =  \tilde{Q}_k(x^2),
\end{equation}
for some monic polynomials $\{ \tilde{p}_k \}$, $\{ \tilde{q}_k \}$, $\{ \tilde{P}_k \}$, $\{ \tilde{Q}_k \}$. The biorthogonality 
condition
\begin{equation}\label{pq1}
\int_{-\infty}^\infty w(x) p_j(x) ({\rm sgn} \,(x))^k q_k(|x|^{ \theta}) \, dx = h_j \delta_{j,k},
\end{equation}
as relevant for the determination of the correlations of (\ref{MB-ensembles1}) according to (\ref{correlation-det}),
is automatically satisfied for $j$ and $k$ of the opposite parity, while for the same parity (\ref{pq1}) with the substitutions
(\ref{pq}) determines
$\{ \tilde{p}_k, \tilde{q}_k \}$ as a biorthogonal system satisfying
\begin{equation}\label{pq2}
\int_0^\infty x^{\theta/2} w(x^{1/2}) 
 \tilde{p}_j(x)  \tilde{q}_k(x^\theta) \, dx = h_{2j+1} \delta_{j,k},
\end{equation}
and  $\{ \tilde{P}_k, \tilde{Q}_k \}$ as a biorthogonal system satisfying
\begin{equation}\label{pq21}
\int_0^\infty x^{-1/2} w(x^{1/2}) 
 \tilde{P}_j(x)  \tilde{Q}_k(x^\theta) \, dx  = h_{2j} \delta_{j,k}.
\end{equation}
Since with $w(x)$ given by (\ref{w1a}), $ x^{\theta/2} w(x^{1/2})$ and  $x^{-1/2} w(x^{1/2})$ is of the form
(\ref{w1a}), knowledge of the biorthogonal system for the latter completely determines that of the former.
In particular, for the normalisations we have the following specifications.

\begin{prop}
Introduce a notation analogous to that used in Proposition \ref{PA1} to annotate the normalisations $h_k$, $h_k^+$.
We have
\begin{align*}
h_{2k}[|x|^{2c}e^{-x^2}]&=h_k^+[x^{c-1/2}e^{-x}],  &  
h_{2k+1}[|x|^{2c}e^{-x^2}]&=h_k^+[x^{c+\theta/2}e^{-x}] \\
h_{2k}[|x|^{2c}(1-x^2)^\alpha]&=h_k^+[x^{c-1/2}(1-x)^\alpha],  &  
h_{2k+1}[|x|^{2c}(1-x^2)^\alpha]&=h_k^+[x^{c+\theta/2}(1-x)^\alpha] \\
h_{2k}[|x|^{2c}/(1+x^2)^\alpha]&=h_k^+[x^{c-1/2}/(1+x)^\alpha],  & 
h_{2k+1}[|x|^{2c}/(1+x^2)^\alpha]&=h_k^+[x^{c+\theta/2}/(1+x)^\alpha]. 
\end{align*}
\end{prop}

We remark that these inter-relations for $\{h_k\}$ are consistent with the product formula~(\ref{normalisation}) and
the factorisations~(\ref{zw2}).

\subsection{Biorthogonal polynomials of q-type}


The biorthogonal polynomials $\{q_k(x^\theta)\}$ are given by the second of the averages in (44). We have following proposition.

\begin{prop}\label{lemma:q}
Consider the Muttalib--Borodin ensemble~\eqref{MB-ensembles}. The $k$-th biorthogonal polynomial, $q_k(x)$, defined through~\eqref{47}, can also be written as an average over the PDF \eqref{MB-ensembles} with $N=k$, $\theta = 1$, according to
\begin{equation}\label{q-schur}
q_k(x)=\sum_{j=0}^k(-1)^{k-j}\frac{\E^{\theta = 1}[s_{\mu^{(k-j)}}(x_1,\ldots,x_k)]}{\E^{\theta = 1}[s_{\mu^{(0)}}(x_1,\ldots,x_k)]}x^j,
\end{equation}
where $\mu^{(j)}=(\mu_1^{(j)},\ldots,\mu_k^{(j)})$ is the partition with 
\begin{equation}
\mu_l^{(j)}=
\begin{cases}
(\theta-1)(k-l)+\theta, &   l\leq j,\\
(\theta-1)(k-l), &  l>j.
\end{cases}
\end{equation}
\end{prop}

\begin{proof}
Expanding the product in the formula (\ref{47})  for $q_k(x)$ in elementary symmetric functions we have
\begin{equation}\label{proof-lem-q/q}
q_k(x)=
\sum_{j=0}^k(-1)^{k-j}\E[e_{k-j}(x_1^\theta,\ldots,x_k^\theta)]x^j.
\end{equation}
Moreover, it follows from~\eqref{schur} that the expectation of a test function, $f(x_1,\ldots,x_k)$, with respect to the joint density~\eqref{MB-ensembles} with $N = k$ may be written as
\begin{equation}\label{proof-lem-q/E}
\E[f(x_1,\ldots,x_k)]=
\frac{\E^{\theta = 1}[f(x_1,\ldots,x_k)s_{\mu^{(0)}}(x_1,\ldots,x_k)]}{\E^{\theta = 1}[s_{\mu^{(0)}}(x_1,\ldots,x_k)]}.
\end{equation}
The final piece needed is to rewrite the elementary symmetric functions as Schur functions; we have
\begin{equation}\label{proof-lem-q/e}
e_r(x_1^\theta,\ldots,x_k^\theta)=s_{(1^r)}(x_1^\theta,\ldots,x_k^\theta)=\frac{s_{\mu^{(r)}}(x_1,\ldots,x_k)}{s_{\mu^{(0)}}(x_1,\ldots,x_k)}.
\end{equation}
Combining~\eqref{proof-lem-q/q}, \eqref{proof-lem-q/E}, and~\eqref{proof-lem-q/e} completes the proof.
\end{proof}

The value of the ratio of averages can be read off from Proposition \ref{P1} in the Jacobi case, the limit $\alpha_2 \to \infty$ of
Proposition \ref{P1} in the Laguerre case, and Proposition \ref{P2} in the Jacobi prime case to give the explicit form of $q_k(x)$
in all these cases.

\begin{corollary}
We have
\begin{align*}
q_k^{\rm J}(x) & = \sum_{j=0}^k (-1)^{k-j } \binom{k}{j} 
\frac{\Gamma(1+a+b+k+\theta j)}{\Gamma(1+a+b+k+\theta k)}\frac{\Gamma(1+a+\theta k)}{\Gamma(1+a+\theta j)}\, x^j \\
q_k^{\rm L}(x) & = \sum_{j=0}^k (-1)^{k-j} \binom{k}{j} 
\frac{\Gamma(1+a+\theta k)}{\Gamma(1+a+\theta j)}\,x^j \\
q_k^{\rm Jp}(x) & = \sum_{j=0}^k (-1)^{k-j } \binom{k}{j} 
\frac{\Gamma(\beta-k-\theta k-\alpha)}{\Gamma(\beta-k-\theta j-\alpha)}\frac{\Gamma(1+\alpha+\theta k)}{\Gamma(1+\alpha+\theta j)}\,x^j.
\end{align*}
\end{corollary}

\begin{proof}
We will give more details for the Jacobi prime case. The ratio of averages in (\ref{q-schur}) is determined by (\ref{SS}).
By noting from (\ref{Sx1}) that
$$
{s_{\mu^{(k-l)}}((1)^k) \over s_{\mu^{(0)}}((1)^k)} = \binom{k}{l},
$$
and noting too that
\begin{align*}
\prod_{l=1}^k {\Gamma(1 + \alpha +  \mu_l^{(k-j)} + k - l) \over \Gamma(1 + \alpha  + \mu_l^{(0)} +k - l)} & =
{\Gamma(1 + \alpha +  \theta k) \over \Gamma(1 + \alpha +  \theta j) } \\
\prod_{l=1}^k {\Gamma(\beta - k   - \mu_l^{(k-j)  } - (k-l) - \alpha ) \over \Gamma(\beta -  k  - \mu_l^{(0)} -(k-l) - \alpha)} & =
 { \Gamma(\beta-k- \theta k   - \alpha) \over   \Gamma(\beta-k- \theta j - \alpha) }, 
 \end{align*}
we thus have that the ratio is fully determined, and the stated result follows.
\end{proof}

\begin{remark}
In the Laguerre and Jacobi cases, these results can be found in \cite{Ko67} and \cite{MT82} respectively.
\end{remark}

With $\{q_k(x)\}$ so determined for the weights (\ref{w1}) we can use the theory in the paragraph beginning with
(\ref{pq}) to specify this same set of  biorthogonal polynomials but now for the weights (\ref{w1a}) as determined by
the condition (\ref{pq1}).

\begin{corollary}\label{P3.6}
Annotate the notation for $q_k(x)$ by writing $q_k(x;w(x))$ in the case of the weights (\ref{w1a}),
and $q_k^+(x;w(x))$ in the case of the weights (\ref{w1}). We have
\begin{align*}
q_{2k}(x; |x|^{2c} e^{- x^2}) & = q_k^+(x^2; x^{-1/2 + c} e^{-x}) \\
q_{2k}(x; |x|^{2c} (1 - x^2)^\alpha) & = q_k^+(x^2; x^{-1/2 + c} (1 - x)^\alpha) \\
q_{2k}(x; |x|^{2c} /(1 + x^2)^\alpha) & = q_k^+(x^2; x^{-1/2 + c} /(1 + x)^\alpha) \\
q_{2k+1}(x; |x|^{2c} e^{-x^2}) & =  x q_k^+(x^2; x^{\theta/2 + c} e^{-x}) \\
q_{2k+1}(x; |x|^{2c} (1 - x^2)^\alpha) & = x q_k^+(x^2; x^{\theta/2 + c} (1 - x)^\alpha) \\
q_{2k+1}(x; |x|^{2c} /(1 + x^2)^\alpha) & = x q_k^+(x^2; x^{\theta/2 + c} /(1 + x)^\alpha).
\end{align*}
\end{corollary}

\subsection{Biorthogonal polynomials of p-type}

The biorthogonal polynomials $\{p_k(x)\}$ are given by the first of the averages in (\ref{pq}), which is precisely
the averaged characteristic polynomial.
 However, evaluation of this average
using Schur polynomial theory no longer gives the coefficient of $x^j$ as a single ratio of averages of Schur polynomials,
but rather as a multiple sum of such averages.  The mechanism for this is most clearly illustrated by considering the class of
PDFs
\begin{equation}\label{SLa}
{1 \over C_N^{\gamma}[w]} \prod_{l=1}^N w(x_l) \prod_{1 \le i < j \le N} (x_i - x_j)  \det[x_j^{\gamma_i}]_{i,j=1,\dots,N},
\end{equation}
which includes the Jacobi (\ref{MKK5}), Laguerre (\ref{MKK5a}) and Jacobi prime (\ref{MKK5x}) cases.

\begin{prop}\label{Ps}
Denote an average with respect to (\ref{SLa}) by $\langle \cdot \rangle^{(\gamma)}$.
The averaged characteristic polynomial is given in terms of the normalisation $ C_N^{\gamma}[w]$ in 
(\ref{SLa}) according to 
\begin{equation}\label{Syd}
p_k^{\{\gamma\}}(x):=
 \Big \langle \prod_{l=1}^N (x - x_l) \Big \rangle^{(\gamma)} = \sum_{\nu =0}^N (-1)^{N-\nu} x^\nu
 \sum_{1 \le l_1 < \cdots < l_{N-\nu} \le N} 
{ C_N^{\gamma^{\rm l}}[w] \over C_N^{\gamma}[w]} ,
\end{equation}
 where $\gamma^{\rm l} = (\gamma_1^{\rm l}, \dots, \gamma_N^{\rm l})$ with
\begin{equation}\label{Syd1} 
 \gamma^{\rm l} = \left \{ \begin{array}{ll} \gamma_i + 1, & i \in \{l_1,\dots, l_{N-\nu} \} \\
 \gamma_i, & {\rm otherwise}. \end{array} \right.
\end{equation}
\end{prop}

\begin{proof}
With $e_\nu(x_1,\dots,x_N)$ denoting the $\nu$-th elementary symmetric function, we have the expansion
$$
 \Big \langle \prod_{l=1}^N (x - x_l)  \Big \rangle^{(\gamma)}  = \sum_{\nu=0}^N (-1)^{N-\nu} x^\nu  
  \Big \langle e_{N-\nu}(x_1,\dots,x_N)  \Big \rangle^{(\gamma)};
  $$
  cf.~(\ref{proof-lem-q/q}).
It follows from the definition of the elementary symmetric function that
$$
 e_{N-\nu}(x_1,\dots,x_N)   \det[x_j^{\gamma_i}]_{i,j=1,\dots,N} = \sum_{1 \le l_1 < \cdots < l_{N-\nu} \le N}  \det[x_j^{\gamma_i^{\rm l}}]_{i,j=1,\dots,N}
 $$
(note that this formula is equivalent to the Pieri formula  in the theory of the Schur polynomials,
see e.g.~\cite[Prop.~12.8.4]{Fo10}; another use of the Pieri formula in the evaluation of averages of
characteristic polynomials can be found in \cite{Fo13a}),
so the task is to evaluate
$$
\int_0^\infty dx_1 \cdots \int_0^\infty dx_N \, \prod_{l=1}^N w(x_l)  \prod_{1 \le i < j \le N} (x_i - x_j)  \det[x_j^{{\gamma_i}^{\rm l}}]_{i,j=1,\dots,N}.
$$
But by definition, this is just the normalisation $ C_N^{\gamma^{\rm l}}[w]$, so (\ref{Syd}) follows.
\end{proof}

The general expression the biorthogonal polynomials $\{p_k(x)\}$ as appearing in~(\ref{pq}) can be read off directly from Proposition~\ref{Ps} by an appropriate choice for $\{\gamma\}$. 
Furthermore, in many cases it possible to simplify this expression by reducing the multiple sum in~\eqref{Syd} to a single sum. 
We will first show how this simplification is done in the Laguerre case. Subsequently, we will be able to extend this analysis to the Jacobi and Jacobi prime cases using appropriate modification. 

Recall that in the Laguerre case (\ref{MKK5a}), we have
\begin{equation}\label{n1}
C_N^\gamma[e^{-x}] = N! \prod_{l=1}^N \Gamma(\gamma_l + 1)  \prod_{1 \le i < j \le N} (\gamma_i - \gamma_j).
\end{equation}
Let us introduce a short-hand notation for the multiple sum in (\ref{Syd}),
\begin{equation}\label{Mas}
F_\nu^{\rm L}(\{\gamma_l\}_{l=1}^N) =   \sum_{1 \le l_1 < \cdots < l_{N-\nu} \le N} 
{ C_N^{\gamma^{\rm l}}[e^{-x}] \over C_N^{\gamma}[e^{-x}]}.
\end{equation}
We have the following lemma.


\begin{lemma}\label{p3.7}
The multiple sum~\eqref{Mas} satisfies the initial condition
\begin{equation}\label{Ma}
F_0^{\rm L}(\{\gamma_l\}_{l=1}^N) =  (\gamma_1 + 1)  (\gamma_2 + 1) \cdots  (\gamma_N + 1)
\end{equation}
and the recurrence relation
\begin{equation}\label{Syda}
F_\nu^{\rm L}(\{\gamma_l\}_{l=1}^N) \Big |_{\gamma_1 = -1} = F_\nu^{\rm L}(\{\gamma_l + 1\}_{l=2}^N) .
\end{equation}
It follows from these formulas that
\begin{equation}\label{Sydb}
F_\nu^{\rm L}(\{\gamma_l\}_{l=1}^N) = \sum_{s=0}^{\nu} c_s^{(p)} \prod_{l=1}^N (\gamma_l + s+1), \quad  c_s^{(\nu)}  = {(-1)^{\nu-s} 
\over  (\nu-s)! s!}.
\end{equation}
\end{lemma}

\begin{proof}
According to (\ref{Mas}), in the case $\nu=0$ there is only a single term in the sum, which from (\ref{Syd1}) is equal to
$$
{C_N^\gamma[e^{-x}] \Big |_{\gamma \mapsto \gamma + 1} \over C_N^\gamma[e^{-x}]}.
$$
We now read off from (\ref{n1}) the expression (\ref{Ma}).
Let $\{j_1,\dots,j_\nu\}$ denote $\{1,2,\dots,N\}$ with $\{l_1,l_2,\dots,l_{N-\nu}\}$ as appearing in (\ref{Syd}) removed, where
$ j_1 < j_2 < \cdots <j_\nu$.
We then have from the definition (\ref{Syd1}) and (\ref{n1}) that
\begin{align}\label{Sydc}
{ C_N^{\gamma^{\rm l}}[e^{-x}] \over C_N^{\gamma}[e^{-x}]} & =
\prod_{l = 1 \atop l \ne j_1,\dots, j_{\nu}}^N { ( \gamma_l + 1 - \gamma_{j_1}) \cdots   ( \gamma_l + 1 - \gamma_{j_\nu}) \over
 ( \gamma_l  - \gamma_{j_1}) \cdots   ( \gamma_l  - \gamma_{j_\nu})  } (\gamma_l+1) \nonumber \\
 & = 
\prod_{l = 1 \atop l \ne j_1,\dots, j_{\nu}}^N \Big ( 1 + {1 \over \gamma_l - \gamma_{j_1}} \Big ) \cdots 
\Big ( 1 + {1 \over \gamma_l - \gamma_{j_\nu}} \Big ) (\gamma_l  + 1). 
\end{align}
The recurrence (\ref{Syda}) is immediate.

On the other hand, we must have
\begin{equation}\label{62}
F_\nu^{\rm L}(\{\gamma_l\}_{l=1}^N) =
\sum_{s=0}^{N} c_s^{(\nu)} \prod_{l=1}^N ( \gamma_l + s + 1)
\end{equation}
for some $ c_s^{(\nu)}$ independent of $\{  \gamma_l \}$. The justification for this is that
$F_\nu^{\rm L}(\{\gamma_l\}_{l=1}^N)$ is, from the second equality in (\ref{Sydc}),
a symmetric polynomial in $\{\gamma_l\}_{l=1}^N$ with highest degree 
term no bigger than $e_{N}(\{ \gamma_l + 1\}_{l=1}^N )$, and $\{ \prod_{l=1}^N ( \gamma_l +1 +  s) \}_{s=0}^{N}$ is a basis for symmetric
polynomials in $\{ \gamma_l \}_{l=1}^N $ with highest degree term $e_{N}(\{ \gamma_l +1\}_{l=1}^N )$.

Substituting (\ref{62}) in the functional equation (\ref{Syda}) shows
\begin{equation}\label{63}
\sum_{s=0}^{N-1} s  c_s^{(\nu)} \prod_{l=1}^N ( \gamma_l + k) = \sum_{s=0}^{N-1}  c_s^{(\nu-1)}
 \prod_{l=2}^{N} ( \gamma_l + s + 1).
\end{equation} 
 Equating coefficients of $\{ \prod_{l=1}^N ( \alpha_l + s) \}_{s=0}^{N-1}$ shows that
 $s c_s^{(\nu)} = c_{s-1}^{(\nu-1)}$. Furthermore, since from the second equality
 in (\ref{Sydc}) the coefficient of 
 $\prod_{l=1}^N (\gamma_l + 1)$ is zero for $\nu \ge 1$, we must then have $\sum_{s=0}^{N}   c_s^{(\nu)}  =  0$, while for
 $\nu=0$, $c_0^{(0)} = 1$,  $c_s^{(0)} = 0$ $(s=1,2,\dots)$. These properties together uniquely determine $ c_s^{(\nu)}$
 as given in   (\ref{Sydb}).
\end{proof}

The expression for the averaged characteristic polynomials in the Laguerre case follows directly from Proposition~\ref{Ps} and Lemma~\ref{p3.7}.

\begin{corollary}
The averaged characteristic polynomials for the Laguerre case of (\ref{SLa}) with $N=k$ are given by
\begin{equation}\label{Syde}
p_k^{{\rm L}  \, \{\gamma\}}(x)= \sum_{\nu = 0}^k  x^\nu  \sum_{s=0}^\nu {(-1)^{k - s} \over (\nu - s)! s!}
 \prod_{l=1}^k (\gamma_l + s+1).
\end{equation}
\end{corollary}

In order to extend the above derivation beyond the Laguerre case, we first note that the quantity (\ref{Mas}) is equal to $\langle e_\nu(x_1,\dots,x_N) \rangle^{(\gamma)}$. It turns out that the recurrence~(\ref{Syda}) holds independent of the weight function $w(x)$ as long as $w(0) \ne 0$.

\begin{lemma}
Suppose $w(0) \ne 0$. We have
\begin{equation}\label{dd}
\lim_{\gamma_N \to -1} \langle e_\nu(x_1,\dots,x_N) \rangle^{(\gamma)} = \langle e_{\nu}(x_1,\dots,x_{N-1}) \rangle^{(\gamma+1)},
\end{equation}
where the average on the right-hand side is with respect to the PDF (\ref{SLa})  with $N \mapsto N - 1$.
\end{lemma}

\begin{proof}
It is generally true that
\begin{equation}\label{tf}
\lim_{\gamma_N \to -1} (1 + \gamma_N) \int_0^\infty t^{\gamma_N} f(t) \, dt = f(0),
\end{equation}
see e.g.~\cite[Prop.~4.1.3]{Fo10}; this identity plays a key role in Selberg's \cite{Se44} original proof of
(\ref{3.2}). Expanding the factor $\det [ x_j^{\gamma_i} ]_{i,j=1}^N$ in both the numerator and
denominator of $\langle e_\nu(x_1,\dots,x_N) \rangle^{(\gamma)}$ by the first row, and then applying (\ref{tf}) with $t=x_N$ and
\begin{multline}
f(t)=h_\nu(t)= w(t) \int_{[0,\infty)^{N-1} }e_\nu(x_1,\dots,x_{N-1},t)
\\ \times \prod_{l=1}^{N-1} w(x_l) (x_l - t)
\!\!\!\prod_{1\le j<k\le N-1}\!\!\! (x_j-x_k)\det[x_j^{\gamma_k} ]_{j,k=1}^{N-1}
dx_1\cdots dx_{N-1}
\end{multline}
in the numerator, and $f(t) = h_0(t)$ in the denominator, (\ref{dd}) results.
\end{proof}

Now, examination of the proof of Lemma~\ref{p3.7} shows that the simplicity of the formula (\ref{Sydb}) relies not only on the
recurrence (\ref{Syda}) but also on the factorised polynomial form of the initial condition (\ref{Ma}).
For the Jacobi weight, it follows from (\ref{MKK5}) and (\ref{MKK6}) that the right-hand side of the evaluation formula
(\ref{Ma}) needs to be multiplied by 
$1/\prod_{l=1}^N(\gamma_l+N+\alpha_2)$, taking us outside the class of polynomials. However, the simple dependence (\ref{Ma})
on $\{\gamma_l\}$ again holds in the Jacobi case if we first change variables $x_j = X_j/(1 + X_j)$ with $0 < X_j < \infty$ so that
(\ref{MKK5}) reads
\begin{equation}\label{dde}
{1 \over N!\,C_N^{\rm J}} \prod_{1 \le i < j \le N}{1 \over  (\gamma_i - \gamma_j)}
\prod_{l=1}^N {1 \over (1 + X_l)^{\alpha_2 + N }}
\prod_{1 \le j < k \le N} (X_k- X_j) \det \Big [
\Big ( {X_j \over 1 + X_j} \Big )^{\gamma_i} \Big ]_{i,j=1}^N.
\end{equation}
(We remark that this functional form appears in \cite{FW15} as the eigenvalue PDF for random matrices $(Y^\dagger Y)^{-1} Z^\dagger Z$,
with $Y$ and $Z$ specified as below (\ref{MKK6}).)
Specifically, with respect to the PDF (\ref{dde}) we have
\begin{equation}\label{62i}
\Big \langle \prod_{l=1}^N X_l \Big \rangle^{(\gamma)} =  
\frac{C_N^{\rm J}\big\vert^{\gamma \mapsto \gamma + 1}_{\alpha_2 \mapsto \alpha_2 - 1}}{C_N^{\rm J}}  =
{\Gamma(\alpha_2-1) \over \Gamma(N+\alpha_2-1)} \prod_{l=1}^N (\gamma_l + 1)
\end{equation}
On the other hand, we can check that the recurrence (\ref{dd}) remains valid for the PDF (\ref{dde}),
and so analogous to (\ref{62}) we must have
\begin{equation}\label{62a}
\langle e_\nu(X_1,\dots,X_N) \rangle^{(\gamma)} = \sum_{s=0}^N \tilde{c}_s^{(\nu,N)}
\prod_{l=1}^N(\gamma_l+s+1);
\end{equation}
here though the coefficients $\tilde{c}_s^{(\nu,N)}$ depend on $N$, as well as $\nu$, unlike in (\ref{62}).
Substituting in (\ref{62a}) shows
$$
\sum_{s=0}^N s \tilde{c}_s^{(\nu,N)} \prod_{l=1}^N(\gamma_l + k ) =
\sum_{s=0}^{N-1} \tilde{c}_s^{(\nu-1,N-1)} \prod_{l=1}^{N-1} (\gamma_l + k)
$$
and thus $s \tilde{c}_s^{(\nu,N)} = \tilde{c}_s^{(\nu-1,N-1)}$. We again have the additional relation $\sum_{s=0}^N \tilde{c}_s^{(\nu,N)} = 0$
for $\nu \ge 1$ while for $\nu=0$, $\tilde{c}_0^{(0,N)} = \Gamma(\alpha_2-1) / \Gamma(N + \alpha_2-1)$,
$\tilde{c}_s^{(0,N)} = 0$ $(s=1,2,\dots)$. Consequently
\begin{equation}\label{62j}
\tilde{c}_s^{(\nu,N)}  = {\Gamma(\alpha_2-1) \over \Gamma(N-\nu+\alpha_2-1)}
{(-1)^{\nu-s} \over (\nu-s)! s!}, \quad s \le \nu, \qquad \tilde{c}_s^{(\nu,N)} = 0, \quad s > \nu.
\end{equation}
With this information, the explicit form of $\{ p_k^{\rm J}(x) \}$ for the PDF  (\ref{SLa}) with $w(x) = (1 - x)^{\alpha_2 - 1}$
can be determined.

\begin{corollary}\label{P3.9}
The averaged characteristic polynomials for (\ref{SLa}) with $N=k$ and weight $w(x) = (1 - x)^{\alpha_2 - 1}$ are given by
\begin{equation}\label{62y}
 p_k^{{\rm J}  \, \{\gamma\}}(x) = {(1 - x)^k \over \prod_{l=1}^k(\gamma_l+k+\alpha_2)}
  \sum_{\nu = 0}^k {\Gamma(k+\alpha_2) \over \Gamma(k - \nu + \alpha_2)}
 \Big ( {x \over 1 - x} \Big )^\nu \sum_{s=0}^\nu {(-1)^{k - s} \over (\nu - s)! s!}
 \prod_{l=1}^k ( \gamma_l + s + 1).
 \end{equation}
 \end{corollary}
 
 \begin{proof}
 Denote an average with respect to  (\ref{dde}) with $N=k$ by $\langle \cdot \rangle^{(\gamma,{\rm JX})}$, and an average
 with respect to    (\ref{SLa}) with $N=k$ and $w(x) = (1 - x)^{\alpha_2 - 1}$ by
 $\langle \cdot \rangle^{(\gamma, {\rm J})}$. By the change of variables $ X_l = x_l/(1 - x_l)$ we have
 $$
\Big \langle \prod_{l=1}^k \Big ( {y \over 1 - y} - X_l \Big ) \Big \rangle^{(\gamma,{\rm JX})} =
{
C_k^{\rm J} \big |_{\alpha_2 \mapsto \alpha_2 - 1} \over C_k^{\rm J} (1 - y)^k} \Big \langle \prod_{l=1}^k (y - x_l) \Big \rangle^{(\gamma,{\rm J})} 
\Big |_{\alpha_2 \mapsto \alpha_2 - 1}.
$$
Rearranging this equation and making use of (\ref{MKK6}), (\ref{62a}) and (\ref{62j}), with the former summed over $\nu$, gives
(\ref{62y}). 
\end{proof}

The Jacobi prime case yields to a similar approach. In the PDF~(\ref{MKK5x}) we change variables $x_j = X_j/(1 - X_j)$ with $0 < X_j < 1$,
 to obtain
\begin{equation}\label{MKK5xa}
{1 \over N!\,C_N^{\rm Jp}}\prod_{1 \le i < j \le N}{1 \over  (\gamma_i - \gamma_j)}
\prod_{l=1}^N  (1 - X_l)^{d +  N - 1} 
\prod_{1 \le i < j \le N} (X_i - X_j)
\det \Big [ \Big ( {X_j \over 1 - X_j}   \Big )^{\gamma_i} \Big ]_{i,j=1}^{N}.
\end{equation}
Analogous to (\ref{62i}) we have
\begin{equation}\label{77}
\Big \langle \prod_{l=1}^N X_l \Big \rangle^{(\gamma)} =    
\frac{C_N^{\rm Jp} \big\vert^{\gamma \mapsto \gamma + 1}_{d \mapsto d + 1}}{C_N^{\rm Jp}}  =
{\Gamma(d + N + 1) \over \Gamma(d + 2N+1)} \prod_{l=1}^N (\gamma_l + 1),
\end{equation}
while the recurrence (\ref{dd}) remains valid. Setting $d+ 2N = \beta$, these two facts allow us to deduce that
\begin{equation}\label{72a}
\langle e_\nu(X_1,\dots,X_N) \rangle^{(\gamma)} = \sum_{s=0}^N  {\alpha}_s^{(\nu,N)}
\prod_{l=1}^N(\gamma_l+s+1),
\end{equation}
with 
\begin{equation}\label{72j}
{\alpha}_s^{(\nu,N)}  = {\Gamma(\beta - N + 1 + \nu ) \over \Gamma(\beta + 1)}
{(-1)^{\nu-s} \over (\nu-s)! s!}, \quad s \le \nu, \qquad \tilde{c}_s^{(\nu,N)} = 0, \quad s > \nu
\end{equation}
(cf.~(\ref{62a}) and (\ref{62j})). Repeating now the working of the proof of Corollary~\ref{P3.9} gives the
explicit form of $\{ p_k^{\rm Jp} \}$ for the PDF (\ref{MKK5x}).

\begin{corollary}\label{P3.10}
The averaged characteristic polynomials for (\ref{SLa}) with $N=k$ and weight $w(x) =1/(1+x)^{\beta}$ are given by
\begin{align}\label{62x}
 p_k^{{\rm Jp} \, \{\gamma\}} (x)  = 
     { (1 + x)^k  \over \prod_{l=1}^k(\beta - k - \gamma_l - 1)} 
  \sum_{\nu = 0}^k \frac{\Gamma(\beta- k  + \nu )}{\Gamma(\beta - k)} 
 \Big ( {x \over 1 + x} \Big )^\nu  \sum_{s=0}^\nu {(-1)^{k - s} \over (\nu - s)! s!}
 \prod_{l=1}^k ( \gamma_l + s + 1).
 \end{align}
 \end{corollary}
 
 Knowledge of the averaged characteristic polynomial for the PDF  (\ref{SLa})  with $N = k$ and $w(x) = e^{-x}, \,
 (1 - x)^{\alpha_2 - 1}, \, 1/(1 + x)^{\beta}$ as given by (\ref{Syde}),  (\ref{62y}) and (\ref{62x}) respectively allows
 us to give explicit formulas for the biorthogonal polynomials $\{ p_k(x) \}$ as specified by (\ref{47})  with Laguerre,
 Jacobi and Jacobi primes weights respectively.
 
 \begin{corollary}
 Consider the Muttalib--Borodin ensemble (\ref{MB-ensembles}) with $N= k$ and  the weights (\ref{w1}). 
 The biorthogonal polynomials $\{ p_k(x) \}$ as specified by (\ref{47}) are given by
 \begin{align}
 p_k^{\rm L}(x) & =  \theta^k \sum_{\nu = 0}^k x^\nu  \sum_{s=0}^\nu {(-1)^{k- s} \over (\nu - s)! s!}
 {\Gamma(k + (a + s + 1)/\theta) \over \Gamma( (a + s + 1)/\theta) } \\
  p_k^{{\rm J} }(x)  & =    (1 - x)^k {\Gamma((a+ b + k + 1)/\theta) \over \Gamma(k + (a+b+k+1)/\theta)}  \nonumber \\
  & \times \sum_{\nu = 0}^k {\Gamma(k+b+1) \over \Gamma(k - \nu + b + 1)}  \Big ( {x \over 1 - x} \Big )^\nu
 \sum_{s=0}^\nu {(-1)^{k - s} \over (\nu - s)! s!}
 {\Gamma(k + (a + s + 1)/\theta) \over \Gamma( (a + s + 1)/\theta) }    \\
  p_k^{{\rm Jp}} (x)  & =  (1 + x)^k  {1 \over \Gamma(\beta - k)}
   {\Gamma((\beta - k - \alpha - 1)/\theta - k + 1) \over  \Gamma((\beta - k - \alpha - 1)/\theta  + 1) } \nonumber \\
 & \times
  \sum_{\nu = 0}^k \Gamma(\beta - k  - \nu ) 
 \Big ( {x \over 1 + x} \Big )^\nu  \sum_{s=0}^\nu {(-1)^{k - s} \over (\nu - s)! s!}
  {\Gamma(k + (\alpha + s + 1)/\theta) \over \Gamma( (\alpha + s + 1)/\theta) } .
 \end{align}
 \end{corollary}
 
 \begin{proof}
 In  (\ref{Syde}) we substitute $\gamma_j = \theta(j-1) + a$. We substitute the same in (\ref{62y}) and also set $\alpha_2 = b+1$,
 while in  (\ref{62x})  we substitute $\gamma_j = \theta(j-1) + \alpha$ .
 \end{proof}
 
 \begin{remark}
In the Laguerre and Jacobi cases, these formulas (up to the normalisation, which we have specified by requiring that
the polynomials be monic, and the replacement $x \mapsto (1-x)/2$ in the Jacobi case) can be found in \cite{Ca68}
and \cite{MT82} respectively.
\end{remark}

The theory in the paragraph beginning with
(\ref{pq}) allows us to specify the biorthogonal polynomials $\{p_k(x)\}$ in the case of the weights (\ref{w1a}) in
terms of this families  $\{p_k(x)\}$  for the weights  (\ref{w1}) as just determined.

\begin{corollary}
Annotate the notation for $p_k(x)$ by writing $p_k(x;w(x))$ in the case of the weights (\ref{w1a}),
and $p_k^+(x;w(x))$ in the case of the weights (\ref{w1}). The formulas of Corollary \ref{P3.6} again hold
with each $q_j$ replaced by $p_j$, and each $q_j^+$ replaced by $p_j^+$.
\end{corollary}
 
 \begin{remark}
 The biorthogonal system  specified by (\ref{pq1}) in the case of the generalised Gaussian weight in (\ref{w1a})
 has previously been studied in \cite{TM86}.
 \end{remark}

 \subsection*{Acknowledgements}
 This work was supported by the Australian Research Council, through grant DP170102028, 
and through the  Centre of Excellence for Mathematical and Statistical Frontiers.


\begin{thebibliography}{10}

\bibitem{AV83}
W.A. Al-Salam and A.~Verma, \emph{$q$-{K}onhauser polynomials}, Pacific J.
  Math. \textbf{108} (1983), 1--7.

\bibitem{BR93}
C.W.J. Beenakker and B.~Rejaei, \emph{Nonlogarithmic repulsion of transmission
  eigenvalues in a disordered wire}, Phys. Rev. Lett. \textbf{71} (1993),
  3689--3692.

\bibitem{BR94}
C.W.J. Beenakker and B.~Rejaei, \emph{Random-matrix theory of parametric
  correlations in the spectra of disordered metals and chaotic billiards},
  Physica A \textbf{203} (1994), 61--90.

\bibitem{Bor99}
A.~Borodin, \emph{Biorthogonal ensembles}, Nucl. Phys. B \textbf{536} (1998),
  704--732.

\bibitem{Ca68}
L.~Carlitz, \emph{A note on certain biorthogonal polynomials}, Pacific J. Math.
  \textbf{24} (1968), 425--430.

\bibitem{Ch14}
D.~Cheliotis, \emph{Triangular random matrices and biothogonal ensembles},
  arXiv:1404.4730, 2014.

\bibitem{CR14}
T.~Claeys and S.~Romano, \emph{Biorthogonal ensembles with two-particle
  interactions}, Nonlinearity \textbf{27} (2014), 2419--2444.

\bibitem{De86}
J.~Deruyts, \emph{Sur une class de polyn\^omes conjug\`es}, Mem. Cor. et Mem.
  de Savant Enstr., Acad. Royal des Sci. des letters et des Beaux, Art de
  Belgique \textbf{48} (1886).

\bibitem{Di69}
F.~Didon, \emph{Sur certains syst\`emes de polyn\^pmes associ\'es}, Annal Sci
  de l'Ecole Normale Sup \textbf{6} (1869), 111--125.

\bibitem{EK95a}
B.~Eynard and C.~Kristjansen, \emph{Exact solution of the {O}(n) model on a
  random lattice}, Nucl. Phys. B \textbf{455} (1995), 577--618.

\bibitem{EZ92}
B.~Eynard and J.~Zinn-Justin, \emph{The {O}(n) model on a random surface:
  critical points and large-order behaviour}, Nucl. Phys. B \textbf{386}
  (1992), 558--591.

\bibitem{FS09}
Z.M. Feng and J.P. Song, \emph{Integrals over the circular ensembles relating
  to classical domains}, J. Phys. A \textbf{42} (2009), 325204.

\bibitem{Fo93a}
P.J. Forrester, \emph{The spectrum edge of random matrix ensembles}, Nucl.
  Phys. B \textbf{402} (1993), 709--728.

\bibitem{Fo06}
\bysame, \emph{Evenness symmetry and inter-relationships between gap
  probabilities in random matrix theory}, Forum Math. \textbf{18} (2006),
  711--743.

\bibitem{Fo10}
\bysame, \emph{Log-gases and random matrices}, Princeton University Press,
  Princeton, NJ, 2010.
  
  \bibitem{Fo13a}
\bysame, \emph{Skew orthogonal polynomials for the real and quaternion real {G}inibre
ensembles and generalizations}, J.~Phys. A \textbf{46} (2013), 245203.

\bibitem{Fo14}
\bysame, \emph{Eigenvalue statistics for product complex {W}ishart
  matrices}, J.~Phys. A \textbf{47} (2014), 345202.

\bibitem{FLZ15}
P.J. Forrester, D.-L. Liu, and P.~Zinn-Justin, \emph{Equilibrium problems for
  {R}aney densities}, Nonlinearity \textbf{28} (2015), 2265--2277.

\bibitem{FL14}
P.J. Forrester and D.-Z. Liu, \emph{Raney distributions and random matrix
  theory}, J. Stat. Phys. \textbf{158} (2015), 1051--1082.

\bibitem{FR02b}
P.J. Forrester and E.M. Rains, \emph{Interpretations of some parameter
  dependent generalizations of classical matrix ensembles}, Prob. Theory
  Related Fields \textbf{131} (2005), 1--61.

\bibitem{FW15}
P.J. Forrester and D.~Wang, \emph{Muttalib--Borodin ensembles in random matrix
  theory --- realisations and correlation functions}, arXiv:1502.07147.

\bibitem{FW07p}
P.J. Forrester and S.O. Warnaar, \emph{The importance of the {S}elberg
  integral}, Bull. Am. Math. Soc. \textbf{45} (2008), 489--534.

\bibitem{FK07}
Y.V. Fyodorov and B.A. Khoruzhenko, \emph{On absolute moments of characteristic
  polynomials of a certain class of complex random matrices}, Comm. Math. Phys.
  \textbf{273} (2007), 561--599.

\bibitem{Ko67}
J.D.E.~Konhauser, \emph{Biorthogonal polynomials suggested by the {L}aguerre
  polynomials}, Pacific J. Math. \textbf{21} (1967), 303--314.

\bibitem{Ka97g}
K.W.J.~Kadell, \emph{The {Selberg-Jack} symmetric functions}, Adv. Math.
  \textbf{130} (1997), 33--102.

\bibitem{Ka93}
J.~Kaneko, \emph{Selberg integrals and hypergeometric functions associated with
  {Jack} polynomials}, SIAM J. Math Anal. \textbf{24} (1993), 1086--1110.

\bibitem{KS97}
F.~Knop and S.~Sahi, \emph{A recursion and combinatorial formula for {J}ack
  polynomials}, Inv. Math. \textbf{128} (1997), 9--22.

\bibitem{KS14}
A.B.J. Kuijlaars and D.~Stivigny, \emph{Singular values of products of random
  matrices and polynomial ensembles}, Random Matrices: Theory and Appl.
  \textbf{3} (2014), 1450011 (22 pages).

\bibitem{KK09}
Y.~Kuramoto and Y.~Kato, \emph{Dynamics of one-dimensional quantum systems},
  CUP, Cambridge, 2009.

\bibitem{LSZ06}
T.~Lueck, H.-J. Sommers, and M.~R. Zirnbauer, \emph{Energy correlations for a
  random matrix model of disordered bosons}, J. Math. Phys. \textbf{47} (2006),
  103304?24.

\bibitem{Ma87}
I.G. Macdonald, \emph{Commuting differential operators and zonal spherical
  functions}, Algebraic Groups, Utrecht 1986 (A.M.~Cohen et~al., ed.), Lecture
  Notes in Math., vol. 1271, Springer, Heidelberg, 1987, pp.~189--200.

\bibitem{MT82}
H.C. Madhekar and N.K. Thakare, \emph{Biorthogonal polynomials suggested by the
  {J}acobi polynomials}, Pacific J. Math. \textbf{100} (1982), 417--424.

\bibitem{MPK88}
P.A. Mello, A.~Pereyra, and N.~Kumar, \emph{Macroscopic approach to
  multichannel disordered conductors}, Ann. Phys. \textbf{181} (1988),
  290--317.

\bibitem{Mu95}
K.A. Muttalib, \emph{Random matrix models with additional interactions}, J.
  Phys. A \textbf{28} (1995), L159?L164.

\bibitem{Se44}
A.~Selberg, \emph{Bemerkninger om et multipelt integral}, Norsk. Mat. Tidsskr.
  \textbf{24} (1944), 71--78.
  
  \bibitem{TM86}
  N.K. Thakare   and  H.C. Madhekar , \emph{Biorthogonal polynomials suggested by the
  {H}ermite polynomials}, Indian J. Pure Appl. Math. \textbf{17} (1986), 1031--1041.

\bibitem{Wa05}
S.O. Warnaar, \emph{$q$-{S}elberg integrals and {M}acdonald polynomials}, The
  Ramanujan J. \textbf{10} (2005), 237--268.

\bibitem{WikiBp}
Wikipedia, \emph{Beta prime distribution},
  https://en.wikipedia.org/wiki/Beta$\_$prime$\_$distribution.

\end{thebibliography}

\providecommand{\bysame}{\leavevmode\hbox to3em{\hrulefill}\thinspace}
\providecommand{\MR}{\relax\ifhmode\unskip\space\fi MR }
\providecommand{\MRhref}[2]{%
  \href{http://www.ams.org/mathscinet-getitem?mr=#1}{#2}
}
\providecommand{\href}[2]{#2}

\end{document}